\documentclass[letterpaper,11pt,prl,reprint,aps,superscriptaddress,showpacs,nofootinbib]{revtex4-2}

\usepackage[utf8]{inputenc}
\usepackage[colorlinks=true,linkcolor=blue,citecolor=blue,urlcolor=blue]{hyperref}
\usepackage{graphicx}
\usepackage{amsmath,amssymb}
\usepackage{dsfont}
\usepackage{xcolor}
\usepackage{amsthm}
\usepackage{float}

\usepackage{algorithm}
\usepackage{algpseudocode}

\algnewcommand\algorithmicforeach{\textbf{for each}}
\algdef{S}[FOR]{ForEach}[1]{\algorithmicforeach\ #1\ \algorithmicdo}

\usepackage{braket}

\newtheorem{lemma}{Lemma}

\newtheorem{theorem}{Theorem}
\newtheorem{corollary}{Corollary}
\theoremstyle{remark}

\definecolor{FGreen}{RGB}{1,100,10}

\newcommand{\idty}{\mathds{1}}

\begin{document}

\title{Partial Self-Correction in Layer Codes}
\author{Dominic J. Williamson}
\affiliation{School of Physics, University of Sydney, Sydney, NSW 2006, Australia}
\date{October 2025}

\begin{abstract}
\noindent
The storage of large-scale quantum information at finite temperature requires an autonomous and reliable quantum hard drive, also known as a self-correcting quantum memory. 
It is a long-standing open problem to find a self-correcting quantum memory in three dimensions. 
The recently introduced Layer Codes achieve the best possible scaling of code parameters and logical energy barrier in three dimensions, these are tantalizing features for the purposes of self-correction. 
In this work we show that a family of Layer Codes, based on good Quantum Tanner Codes, exhibit partial self-correction. Their memory time grows exponentially with linear system size, up to a length scale that is exponential in the inverse temperature. 
At this length scale, the memory time scales as a double exponential of inverse temperature. 
To establish this result we introduce a concatenated matching decoder that combines three rounds of parallelized minimum-weight perfect-matching with a decoder for good Quantum Tanner Codes. We show that our decoder corrects errors up to a constant fraction of the energy barrier, and a constant fraction of the code distance, for a family of Layer Codes. 
Our results position Layer Codes as the leading candidate for a partially self-correcting memory in three dimensions. While they fall short of achieving strict self-correction in the thermodynamic limit, our work highlights the potential of these local codes in three dimensions, with fast distance and logical qubit growth, fast decoders, and a long memory time over a wide range of parameters. 
\end{abstract}

\maketitle

\vspace{1cm}

Isolated quantum information is inherently vulnerable to corruption due to interaction with the environment. Storing quantum information in a scalable way requires protection from noise. This can be achieved with a quantum error-correcting code. The leading approaches to quantum error correction rely on persistent active measurements and feed-forward operations to suppress errors in a scalable quantum memory~\cite{Acharya2023,Bluvstein2023,Lacroix2024,Ryan-Anderson2024}. This leads to a large classical control overhead, which poses a challenge to scaling up fault-tolerant quantum computers. 

The Platonic ideal of a quantum hard drive is a system into which a quantum state can be encoded, left for an arbitrarily long time, and then efficiently decoded to recover the initial state. A self-correcting quantum memory is the next best thing~\cite{Brown2016Quantum}. It is a reasonable physical system that can reliably store encoded quantum information for a time that grows with the system size, while being weakly coupled to a thermal bath. For a more precise definition, we refer the interested reader to Brell's ``Caltech rules''~\cite{Brell2016A}. Self-correcting quantum memories have attracted much recent interest~\cite{Zhu2022Topological,Dua2023Quantum,Miguel2023A,Hong2024Quantum,Hsin2024Non,Zhao2024Energy,Placke2024Topological,Ma2025Circuit,Guo2025Towards,Lin2024Proposals,Roberts2025Cored,Sriram2025Diffusion}. 
In this work, we focus on the specific notion of a self-correcting memory at finite temperature. However, we remark that interesting progress has been made on related notions that allow more complicated local dynamics~\cite{Balasubramanian2024A}. 

The prototypical example of a self-correcting classical memory is the classical Ising ferromagnet in two-dimensions, which can faithfully store a classical bit for a time that is exponential in the linear system size. Here, encoding sends the states of a logical bit to the distinct ground states of the ferromagnet, and decoding is achieved via a majority vote. The canonical example of a self-correcting quantum memory is the four-dimensional toric code with membrane-like logical operators~\cite{dennis2002topological,alicki2010thermal}, which has been explored more recently in Refs.~\cite{Aasen2025A,Bergamaschi2025Rapid}. This brings us to the long standing open question that motivates this work: \textit{Is there a self-correcting quantum memory in three dimensions?} Several necessary conditions have been established, including the presence of a sufficiently strong (free) energy barrier~\cite{Temme2014Thermalization,Temme2015How,Baspin2025The,Rakovszky2024Bottlenecks}. The focus on three dimensions is because in two dimension, or less~\cite{Baspin2025Stabilizer}, self-correction has been ruled out by a constant upper bound on the attainable energy barrier~\cite{Alicki2009,bravyi2009no,Landon-Cardinal2012Local,Komar2016Necessity}. 

The search for a three-dimensional self-correcting quantum memory has motivated numerous interesting proposals. The three-dimensional toric code is not a self-correcting quantum memory at finite temperature due to its deconfined particle excitations~\cite{Castelnovo2007,Castelnovo2008}. However, the confined line excitation sector of the three-dimensional toric code is predicted to make it a stable classical self-correcting memory~\cite{Poulin2018Self}.
Haah made a landmark breakthrough with his discovery of the cubic code~\cite{haah2011fractal}. This model was shown to exhibit a logarithmic energy-barrier and \textit{partial self-correction}, a weak form of self-correcting behaviour with a polynomially growing memory time that only persists up to a temperature dependent system size~\cite{PhysRevLett.107.150504,Bravyi2013}. 
Interestingly, Haah also established a logarithmic upper bound on the energy barrier of any translation-invariant three-dimensional Pauli stabilizer code~\cite{haah2013commuting} which severely limits the achievable memory time scaling. 
Michnicki broke the logarithmic upper bound on the energy barrier in three dimensions by breaking translation symmetry in his construction of welded codes~\cite{Michnicki20123D,Michnicki20143D}. These codes exhibit a polynomial energy barrier $\Delta=\Theta(L^{\frac{2}{3}})$ and a stronger form of partial self-correction with a memory time that grows exponentially with the energy barrier up to a temperature dependent system size~\cite{siva2017topological}. 
Three dimensional topological subsystem codes that support single-shot quantum error correction have also been explored in the context of self-correction~\cite{Bombin2015,Kubica2021Single,Li2024Phase} and are related to the more exotic notion of symmetry-protected self-correction~\cite{Roberts2020Symmetry,Roberts20203F,Stahl2021Symmetry,Stahl2022Self}. 

In this work we assess the performance of the recently introduced \textit{Layer Codes} as self-correcting quantum memories.  Ref.~\cite{Williamson2023} demonstrated that Layer Codes based on any family of good quantum low-density parity-check (qLDPC) codes~\cite{Panteleev2022} achieve the optimal scaling of code parameters $[[\Theta(L^3),\Theta(L),\Theta(L^2)]]$ in three dimensions~\cite{bravyi2010tradeoffs}, see also related works Refs.~\cite{portnoy2023local,Lin2023,Li2024Transform}. Furthermore, it was demonstrated that Layer Codes based on a family of good Quantum Tanner Codes~\cite{leverrier2022quantum} achieve the optimal scaling of the energy barrier $\Delta=\Theta(L)$ in three dimensions. 
A naive estimate of the memory time based on the energy barrier suggests Layer Codes have the potential to be self-correcting. A more careful analysis based on the free energy barrier demonstrates that they are not self-correcting in the strict sense as the system size diverges~\cite{Baspin2025The}. 

Here, we show that Layer Codes based on a family of good Quantum Tanner Codes are in fact partially self-correcting, with a strong growth of the memory lifetime up to a temperature dependent system size. 
We introduce a concatenated matching decoder that combines three rounds of parallel minimum-weight perfect-matching on layers with a decoder for the good Quantum Tanner Codes upon which the Layer Codes are based, see also Ref.~\cite{Gu2025Layer}. 
For any family of Layer Codes based on good Quantum Tanner Codes~\cite{leverrier2022quantum}, the decoders in Refs.~\cite{Leverrier2022Efficient,Dinur2022Good,Gu2022An,Leverrier2022A,Gu2024Single} suffice. 
We show that our decoder is capable of decoding errors up to a constant fraction of the energy barrier and a constant fraction of the distance. 
Combining this result with the techniques from Ref.~\cite{Bravyi2013} we derive the following bound on the memory error $\varepsilon$ after time $t$
\begin{align}
    \varepsilon(t) \leq O(t) \cdot e^{-a c \beta  L }
\end{align}
which implies a memory time scaling 
\begin{align}
    T_{mem} =  \Omega(e^{a c \beta  L}),
\end{align} 
where we have assumed a sufficiently small system size $L\leq O(e^{\frac{1}{3}(1-a)\beta})$ and a sufficiently high inverse temperature $\beta \gg 1 $ for simplicity of presentation. 
When the largest allowable system size is chosen as a function of inverse temperature we find a memory time
\begin{align}
    T_{mem \bullet} &= \Omega\big(\exp\big[ ac \beta e^{\frac{1}{3}(1-a)\beta}\big]\big), 
\end{align}
again for $\beta\gg 1$. 

Our partial-self correction results apply to a number of encoded qubits that grows as fast as possible given the code distance and spatial dimension. 
In summary, we have established a strong form of partial self-correction in Layer Codes with a scaling performance that supersedes previous proposals in the literature.

\begin{figure}[t]
    \centering
    \includegraphics[page=10]{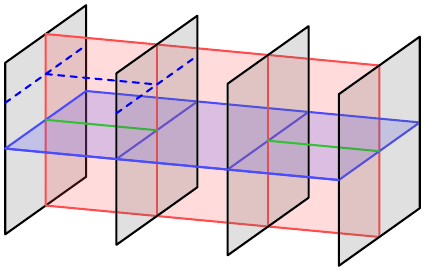}
    \caption{The Layer Code based on the [[4,2,2]] input code.  The grey $xz$-layers, blue $xy$-layer, and red $yz$-layer, depict surface codes that correspond to physical qubits, the $XXXX$ check, and the $ZZZZ$ check, of the input code, respectively. 
    The blue $x$-oriented, red $z$-oriented, and green $y$-oriented, junction lines correspond to nontrivial topological defects, see Fig.~\ref{fig:Defects}. The intersection of two  layers without a junction line indicates a simple decoupled crossing. 
    The grey and red surface codes have rough e-condensing boundaries on the top and bottom, while the grey and blue surface codes have smooth m-condensing boundaries on the front and back. The dashed line indicates an $X$-type logical operator. }
    \label{fig:422}
\end{figure}


\textit{Layer Codes.}---We first briefly introduce relevant background material on Layer Codes~\cite{Williamson2023}, which are a family of stabilizer codes~\cite{gottesman1997stabilizer}. 
The Layer Code construction takes an $[[n,k,d]]$ Calderbank-Shor-Steane~\cite{calderbank1996good,steane1996multiple} (CSS) code as input and outputs a local code in three spatial dimensions. When applied to a qLDPC code, this produces a three-dimensional code with parameters $[[O(n^3),k,\Omega(dn)]]$. Here, for simplicity of presentation, we have assumed the number of $X$-type ($Z$-type) checks in the input code is $\Theta(n)$.  
For an input code with energy barrier $\Delta$, the Layer Code has energy barrier $\Theta(\Delta)$. 
The Layer Codes derived from a family of good Quantum Tanner Codes have code parameters $[[\Theta(L^3),\Theta(L),\Theta(L^2)]]$ and energy barrier $\Delta=\Theta(L)$ which are optimal for a local code in three dimensions~\cite{bravyi2009no,bravyi2010tradeoffs}. 

The Layer Code construction replaces each data qubit, $X$-check, and $Z$-check, of the input code with a patch of surface code~\cite{bravyi1998quantum}. The surface codes assigned to data qubits have standard boundary conditions, while those assigned to $X$-checks ($Z$-checks) have smooth (rough) boundary conditions only. See Fig.~\ref{fig:422} for an example of the Layer Code constructed from the [[4,2,2]] code. The surface code patches are effectively glued together by topological defect lines~\cite{Aasen2020}, which are assigned to a subset of their intersections according to the incidence relations of the associated qubits and checks in the input Tanner graph~\cite{Williamson2023}. Where no defect line is assigned, the intersecting layers pass through one another in a simple decoupled fashion. 
The syndromes of a layer code can be identified locally with the $e$ and $m$ syndromes of each surface code layer.
However, when an error passes through a defect, certain syndromes branch into the intersecting surface code. 
These branching rules underlie the favorable distance and energy barrier scaling properties of Layer Codes.
The $X$-type ($Z$-type) logical operators of a Layer Code are derived from the logical operators of the input code concatenated with the surface code, these are multiplied by additional string operators in the $Z$-check ($X$-check) surface code layers to ensure that they commute with all stabilizers and hence satisfy the branching rules. 
See Fig.~\ref{fig:Defects} for a depiction of the topological defect lines that appear in Layer Codes.  Where the defect lines meet, they define point defects, see Ref.~\cite{Williamson2023} for a detailed description.   

\begin{figure}[t]
    \centering
    \includegraphics[page=5]{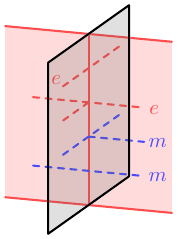} \quad
    \includegraphics[page=6]{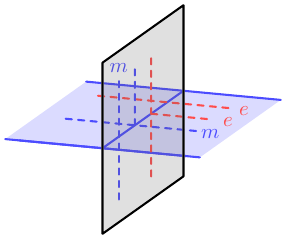}  \quad
    \raisebox{.45cm}{\includegraphics[page=7]{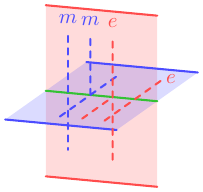}} 
    \caption{The line defects in a Layer Code. Here, we depict generators for the syndromes that can be created or destroyed (condensed) locally at the defect. This is dual to a stabilizer description, which can be found in Ref.~\cite{Williamson2023}. The 3-partite line defects where a red (blue) layer ends can be resolved into the above 4-partite red (blue) defect and a rough e-condensing (smooth m-condensing) boundary, see Fig.~\ref{fig:Boundaries}. }
    \label{fig:Defects}
\end{figure}

\begin{figure}[t]
    \centering
    \includegraphics[page=8]{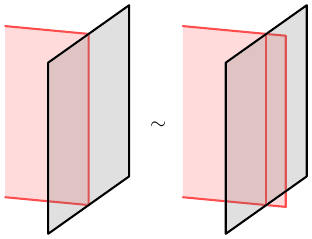} \\
    \includegraphics[page=9]{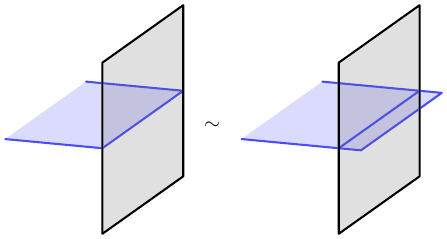}
    \caption{The 3-partite line defects where a red $yz$ (blue $xy$) layer ends can be resolved into a 4-partite defect from Fig.~\ref{fig:Defects} and a smooth e-condensing (rough m-condensing) boundary.}
    \label{fig:Boundaries}
\end{figure}


\textit{Partial Self-Correction.}---We now briefly review relevant background material on partial self-correction. See the appendix for more details on the general setting. 
In Ref.~\cite{Bravyi2013} the authors derive a general bound on the storage error of a qLDPC stabilizer code Hamiltonian interacting with a thermal bath in the Davies weak coupling limit, we follow their notation below. The dynamics is modelled by a standard Markovian master equation that includes a Lindblad generator term to describe the dissipation of energy. The coupling of the system to the bath is described by a set of local quantum jump operators with transition rate coefficients that satisfy detailed balance. This condition biases the dynamics against transitions that increase the total energy, which underlies the utility of the energy barrier for protecting a quantum memory. A constant upper bound on the strength of the system-bath coupling operators, the transition rate coefficients, and the energy penalty of the terms that appear in a Pauli basis expansion of the quantum jump operators is required. 
The logical storage error $\varepsilon$ is quantified by the distance between the initial code state and the code state at time $t$, after the application of a decoder recovery map $\Phi_{\text{ec}}$. 
The bound on this storage error derived in Ref.~\cite{Bravyi2013} is
\begin{align}
    \varepsilon (t) :&= \| \rho(0) - \Phi_{\text{ec}}\big(\rho(t)\big) \|
    \\
    &\leq O(tN) \text{tr} ( Q_m e^{-\beta H} ) \\
    &\leq O(tN) 2^k e^{- a \beta m} ( 1+ e^{-(1-a)\beta} )^N ,
    \label{eq:prePSC}
\end{align}
for any choice of $0<a<1$. Here, $N=\Theta(L^3)$ is the number of checks, $k$ is the number of logical qubits, $m$ is the energy barrier achieved by the decoder, and $Q_m$ is the projection onto states with energy above $m$. The factor $O(tN)$ measures the spacetime volume in which potential syndromes can appear. The remaining factors to the right of the inequality measure a proxy for the error suppression due to the free energy barrier achieved by the decoder. 
A key observation in Ref.~\cite{Bravyi2013} is that the ``entropy'' term $( 1+ e^{-(1-a)\beta} )^N$ is upper bounded by a constant for $N\leq e^{(1-a)\beta}$, hence we have
\begin{align}
\label{eq:PSCBound}
    \varepsilon (t) 
    &\leq O(tN) 2^k e^{- a \beta m} ,
\end{align}
up to sizes $N\leq e^{(1-a)\beta}$. 
This bound implies the following memory time
\begin{align}
    T_{mem} \geq \Omega(e^{a\beta m - k \log 2 - 3\log L })
    \label{eq:TmemGeneral}
\end{align}
up to sizes $L\leq  O(e^{\frac{1}{3}(1-a)\beta})$.
The above bounds were derived and used in Ref.~\cite{Bravyi2013} to establish partial self-correction in the cubic code, and they are what we use below to establish partial self-correction in Layer Codes based on good Quantum Tanner Codes. 

\textit{Partial Self-Correction in Layer Codes.}--- 
We now present our main result. 
\begin{theorem}[Partial Self-Correction in Layer Codes]
    There exist constants $r,c,$ such that for any inverse temperature $\beta$, any state $\rho(0)$ in a code space from a family of Layer Codes based on good Quantum Tanner Codes, any time $t\geq 0$, 
    and any constant $0<a<1$ we have
    \begin{align}
        \varepsilon(t) = \| \rho(0) - \Phi_{\text{ec}}\big(\rho(t)\big) \| \leq O(t) \cdot e^{-(a c \beta -r \log 2) L +3\log L } , \nonumber
    \end{align}
    for $L\leq L_\bullet := O(e^{\frac{1}{3} (1-a) \beta})$.
    Here, $\Phi_{ec}$ is the decoder described in Algorithm~\ref{alg:CMD}, which runs in polynomial time.
    \label{thm:1}
\end{theorem}
\begin{proof}
Ref.~\cite{Williamson2023} established that for Layer Codes based on good Quantum Tanner Codes, there exist an $r$ such that $k(L)\geq r L$. We also have that $N=\Theta(L^3)$. In Lemma~\ref{Lem:1} we show that there exists a $c>0$ such that $m(L)\geq c L$. Combining these bounds with Eq.~\ref{eq:PSCBound} we obtain the stated result.
\end{proof}
The constants $r,c,$ above quantify the encoding rate and relative energy barrier that are achievable by a family of good Quantum Tanner Codes, respectively. 
\begin{corollary}[Layer Code memory time]
The bound in Theorem~\ref{thm:1} immediately implies a memory time
    \begin{align}
    T_{mem} &\geq \Omega(e^{(a c \beta -r \log 2) L -3\log L }) \\
    &\approx \Omega(e^{a c \beta  L}),
    \end{align}
where we have taken $\beta,L \gg 1 $ on the second line to simplify the expression. 
Taking the system size to scale with inverse temperature as $L(\beta)=L_\bullet$ we find
    \begin{align}
    T_{mem \bullet} 
    &\geq \Omega\big(\exp\big[ ac \beta e^{\frac{1}{3}(1-a)\beta}\big]\big)
    \end{align}
    for $\beta \gg 1$. 
    
\end{corollary}
At this point, we have demonstrated that Layer Codes are partially self-correcting, under a thermalizing local noise model that satisfies the assumptions of Ref.~\cite{Bravyi2013}, provided that the decoder in Algorithm~\ref{alg:CMD} corrects errors up to a constant fraction of the energy barrier, see Lemma~\ref{Lem:1}. 
The partial self-correcting behaviour of Layer Codes is particularly strong in comparison to previous results. This is due to the optimal scaling of the energy barrier in Layer Codes. We now highlight several features of partial self-correction in Layer Codes. 
First, our partial-self correction result applies to a number of encoded qubits that scales linearly with $L$, which is the fastest possible growth given the code distance and spatial dimension. Previous results focused on a constant number of encoded qubits. 
Second, the memory time achieved in the partially self-correcting regime has a better scaling with $L$ than previous results. 
This, in turn, leads to improved scaling with $\beta$ when a sequence of system sizes $L(\beta)=L_\bullet$ are considered.  

Due to the fact that the energy barrier scaling of Layer Codes is optimal, the condition $L_\bullet= O(e^{\frac{1}{3} (1-a) \beta})$ is slightly conservative. 
Returning to Eq.~\eqref{eq:prePSC} we find a looser condition on the maximum system size $L_\bullet$ up to which the energy barrier provides some error suppression 
\begin{align}
   \log(1+e^{-(1-a)\beta}) vL_\bullet^3 \ll ( a c \beta -r \log 2 ) L_\bullet - 3\log L_\bullet.  \nonumber
\end{align}
Here, we have used that there exists a $v$ such that ${N \leq v L^3}$. 
We consider $\beta,L_\bullet \gg 1$ such that ${\log(1+e^{-(1-a)\beta})}\approx e^{-(1-a)\beta}$ and so that we can safely neglect $3L_\bullet^{-1} \log L_\bullet$ and $(ac \beta)^{-1} r\log 2$ on the right of the inequality, then 
\begin{align}
   e^{-(1-a)\beta} vL_\bullet^3 &\ll  a c \beta  L_\bullet  \\
   L_\bullet&\ll \sqrt{\frac{a c \beta }{v}}e^{\frac{1}{2}(1-a)\beta} .
\end{align}
Now, scaling the maximum system size with inverse temperature results in a memory time that scales as
\begin{align}
    T_{mem \bullet} &\geq \Omega(e^{a c \beta L_\bullet }) \\
    &\geq \Omega\big(\exp\big[{\sqrt{{v^{-1}(a c)^3}} \beta{^\frac{3}{2}} e^{\frac{1}{2}(1-a)\beta} }\big]\big)
\end{align}
for $\beta \gg 1$. This is a slight improvement on the more conservative estimates given above.


\begin{figure}[t]
    \centering
    \includegraphics[page=1]{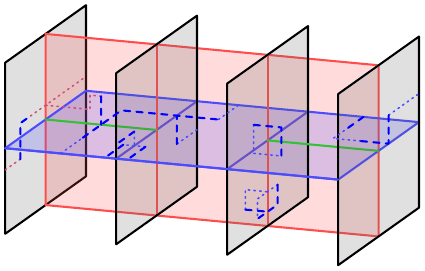}
    \caption{A depiction of an $X$-type error (dashed lines) and a correction operator (dotted lines) on the [[4,2,2]] Layer Code. }
    \label{fig:422error}
\end{figure}

\begin{figure}[t]
    \centering
    \includegraphics[page=2]{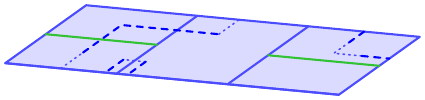} \\
    \includegraphics[page=3]{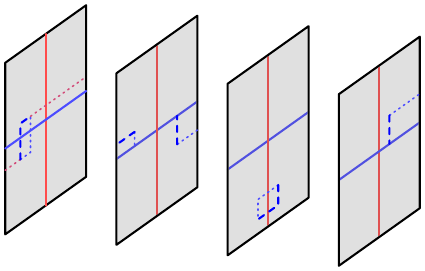} \\ 
    \includegraphics[page=4]{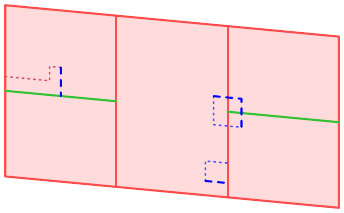}
    \caption{A depiction of the multistage decoder applied to the $X$-type error (dashed lines) on the [[4,2,2]] Layer code in Fig.~\ref{fig:422error}. The errors and corrections (dotted lines) applied at each step can create further syndromes in the subseqeunt steps due to the properties of the defects in Fig.~\ref{fig:Defects}. 
    In the first step, the syndromes (endpoints of dashed lines) on the blue layer are matched using rough boundary conditions on all sides. 
    At the second step, the syndromes in the grey layers (including any from the first correction) are matched (blue dotted lines) using standard surface code boundary conditions. 
    In the third step the parity of syndromes (including any from previous corrections) on the red layer is computed and the input decoder is run on this effective syndrome. The correction returned by the input decoder is applied by flipping the logical sector of the matching on the layers that correspond to qubits flipped by the correction operator. In this example the red layer has negative parity, and so a correction operator $X_0$ is returned. This flips the matching on the leftmost layer (purple dotted lines) to restore even parity on the red layer. 
    In the fourth, final, step, the syndromes on the red layer (including those from previous corrections) are matched assuming smooth boundary conditions (dotted lines). }
    \label{fig:Decoding}
\end{figure}

\begin{figure}[t]
\begin{algorithm}[H]
\caption{Concatenated Matching Decoder}
\label{alg:CMD}
\begin{algorithmic}
    \Require A \textbf{Layer Code} $\mathcal{C}$ based on an input from a family of good Quantum Tanner Codes. 
    \textbf{Syndrome locations} ${S(E)=\{s_i\}}$ created by an $X$-type error $E$. 
    $\textbf{MWPM}$, a minimum-weight perfect-matching decoder that takes in $Z$-check syndrome locations on the surface code and returns a minimum-weight $X$-type correction operator. 
    $\textbf{MWPM}^{-}$, a variant of $\textbf{MWPM}$ that returns a minimum-weight correction in the inequivalent logical sector. 
    $\textbf{QTCD}$, a linear-time decoder for a family of good Quantum Tanner Codes that takes in $Z$-check syndrome locations and returns a list of grey layers corresponding to qubits in the support of an $X$-type correction operator on the input code~\cite{Leverrier2022Efficient,Dinur2022Good,Gu2022An,Leverrier2022A}.
    \vspace{.1cm}
    \Ensure An $X$-type correction operator $R$ that returns $\mathcal{C}$ to the code space.
    \vspace{.1cm}

    \State $R \gets \idty$
    \State $\xi \gets S(E)$
    \State $\sigma(\ell_r)\gets 0$

    \ForEach{blue layer $\ell_b$ in $\mathcal{C}$} 
    \Comment{Parallelizable, stage 1}
    \State $R \gets R \cdot \textbf{MWPM} (\xi \cap \ell_b)$.
    \Comment {Update correction}
    \State $\xi \gets S(E)\cdot S(R) $
    \Comment {$\mathbb{Z}_2$ syndrome update}
    \EndFor
    
    \ForEach{grey layer $\ell_g$ in $\mathcal{C}$} 
    \Comment{Parallelizable, stage 2}
    \State $R \gets R \cdot \textbf{MWPM} (\xi \cap \ell_g)$.
    \State $\xi \gets S(E)\cdot S(R) $
    \EndFor
    
    \ForEach{red layer $\ell_r$ in $\mathcal{C}$} 
    \Comment{Parallelizable, stage 3}
    \State $\sigma(\ell_r)\gets |\xi\cap \ell_r| \mod 2$
    \Comment{Compute parity}
    \EndFor

    \ForEach{grey layer $\ell_g$ in $\textbf{QTCD}(\sigma)$} 
    \Comment{Parallelizable}
    \State $R \gets R \cdot \textbf{MWPM} (\xi \cap \ell_g) \cdot 
    \textbf{MWPM}^- (\xi \cap \ell_g)$
    \State $\xi \gets S(E)\cdot S(R) $
    \EndFor

    \ForEach{red layer $\ell_r$ in $\mathcal{C}$} 
    \Comment{Parallelizable, stage 4}
    \State $R \gets R \cdot \textbf{MWPM} (\xi \cap \ell_r)$
    \State $\xi \gets S(E)\cdot S(R) $
    \EndFor

\end{algorithmic}
\end{algorithm}
\end{figure}

\textit{Concatenated matching decoder.}---
We now describe a concatenated matching decoder for Layer Codes, and argue that it decodes up to a constant fraction of the energy barrier for good Quantum Tanner Code inputs. 
We focus on $X$-type errors, which create $m$ syndromes. The correction of $Z$-type errors proceeds similarly, with the role of $X$-check and $Z$-check layers reversed. 

Our concatenated matching decoder is based on the intuition that it should be possible to decode the individual surface code layers that make up a Layer code. Due to the branching of syndromes at defect lines, see Fig.~\ref{fig:422error}, applying a correction operator to one of the surface code layers can effect the layers that intersect it on defect lines. 
Hence, we need a sequential approach that decodes collections of disjoint layers in each step, taking into account the decoder outcomes from previous steps, similar to Refs.~\cite{Brown2019Parallel,Lee2025Color,Schwartzman-Nowik2025}.

Due to the structure of the defects, there is an asymmetry in the effect of applying a correction. 
For an $X$-type correction, blue layers can be corrected up to branching new syndromes into the grey and red layers. 
Grey layers can be corrected up to branching new syndromes into the red layers. 
There are two sectors of logically inequivalent corrections on each grey layer, due to the standard surface code boundary conditions.
A red layer can only be corrected if its total parity of $m$ syndromes is even. This correction does not branch defects into the blue or grey layers. 
The set of red layers with odd parity is equivalent to a syndrome of the $Z$-type checks in the input code. 
We can find an $X$-type correction for this syndrome on the input code using a decoder for that code. 
To ensure the even parity condition is satisfied on all red layers, we flip the logical sector of the surface code corrections on a set of grey layers. This set corresponds to the qubits in the support of the input code correction. 
It is then possible to correct the red layers, leaving behind no syndromes. 

Since the surface code layers of the same color do not intersect, it is possible to decode them in parallel. 
This leads to a four-stage decoder, depicted in Fig.~\ref{fig:Decoding}: First, the blue surface code layers are decoded and corrected in parallel. Second, the grey layers are decoded and corrected in parallel. Third, the parity of $m$-syndromes on the red layers is computed in parallel and a decoder for $X$-type errors on the input code is called, returning a correction. This is used to flip the logical sector of the correction on each grey layer that corresponds to a qubit in the support of the input code correction. 
Fourth, the red layers are decoded in parallel. At this point the decoder terminates and a full correction has been found.

In Algorithm~\ref{alg:CMD} we describe a specific implementation of the four-stage decoder based on minimum-weight perfect-matching for the surface code layers, and the linear time decoder for good Quantum Tanner Codes from Refs.~\cite{Leverrier2022Efficient,Dinur2022Good,Gu2022An,Leverrier2022A}. 
Here, $S(P)$ denotes the syndrome locations created by an $X$-type pauli operator $P$. 
We remark that the minimum-weight perfect-matching decoder only needs to be called once per grey layer if the result is stored. 
Similarly, the Quantum Tanner Code decoder only needs to be called once and stored. 
The syndrome update steps can be performed in an efficient manner by only modifying the relevant syndromes at each stage.

We now make several remarks about the concatenated matching decoder. 
First, all Pauli correction operators can be implemented classically via Pauli frame tracking. 
Second, the red layer parities computed at stage 3 must correspond to a valid syndrome configuration of the input code, i.e.~a syndrome that satisfies all meta-checks (materialized symmetries)~\cite{kitaev2003fault}. This is because the meta-checks map faithfully from the input code to the Layer Code~\cite{Williamson2023}. 
Third, after the correction operator is applied in stage 3 the parity on each $Z$-check layer is guaranteed to be neutral~\cite{Leverrier2022Efficient,Dinur2022Good,Gu2022An,Leverrier2022A}. Once a minimum-weight correction is returned by $\textbf{MWPM}$, the minimum-weight correction in the inequivalent logical sector $\textbf{MWPM}^-$ can be computed by forcing the parity of matchings to each smooth boundary to be the opposite of those found by $\textbf{MWPM}$. 

Algorithm~\ref{alg:CMD} is highly parallelizable, and has the potential to produce a fast decoder with time complexity
\begin{align}
    t_{\text{CMD}} = O (t_{\text{MWPM}}+t_{\text{QTCD}}) .
\end{align}
The good Quantum Tanner Code decoder has linear time complexity~\cite{Leverrier2022Efficient,Dinur2022Good,Gu2022An,Leverrier2022A} and hence $t_{\text{QTCD}}=O(L)$.
If an implementation of minimum-weight perfect-matching is used that achieves $t_{\text{MWPM}} = O(L^3)$, then the overall decoding time is (sub) linear. 
In particular, the sparse-blossom algorithm was observed to have scaling $t_{\text{MWPM}}=O(L^{2(1+\epsilon)})$ for $0 \leq\epsilon\leq .32$ in Ref.~\cite{Higgott2025Sparse}. This leads to an expected \textit{sublinear scaling} of the time complexity ${t_{\text{CMD}} = O(n^{\frac{2}{3}(1+\epsilon)})}$, where $n=\Theta(L^3)$.

\textit{Decoder energy barrier.}---We now present the main technical result of this work: The energy barrier of the concatenated matching decoder is proportional to the energy barrier of the decoder called for the input qLDPC code. When combined with the linear energy barrier of \textbf{QTCD}, this establishes that the concatenated matching decoder achieves a constant fraction of the energy barrier for the associated Layer Codes. 
\begin{lemma}[Concatenated matching decoder energy barrier]
\label{Lem:1}
    For Layer Codes based on a family of good Quantum Tanner Codes, Algorithm~\ref{alg:CMD} successfully decodes any sequence of Pauli errors with energy penalty up to a constant fraction of the energy barrier. That is, there exists a constant $c>0$ such that $m=c L$. 
\end{lemma}
\begin{proof}
    The energy barrier proof in Ref.~\cite{Williamson2023} established that any sequence of local errors $\{P_0P_1\cdots P_j\}_{j=0,\dots,\tau}$ with energy penalty $\epsilon$ in a Layer Code maps to a sequence of local errors on the input qLDPC code $\{P'_0P'_1\cdots P'_j\}_{j=0,\dots,\tau}$ with energy penalty $\epsilon'$. 
    The error on the input code at step $j$ is found by running Algorithm~\ref{alg:CMD} on the Layer Code with error $P_0P_1\cdots P_j$ and returning the correction operator produced by \textbf{QTCD} at the end of stage 3. 
    In the worst case, this mapping leads to a constant multiplicative increase in the energy penalty of the sequence $\epsilon'\leq \frac{w w' }{4}\epsilon$. Here, $w$ is the weight and $w'$ is the degree of the input code. Hence, the concatenated matching decoder achieves a $\frac{4}{w w'}$ fraction of the energy barrier achieved by \textbf{QTCD}. 

    The good Quantum Tanner Codes exhibit linear confinement around any code state, for errors of weight up to a constant fraction of the distance~\cite{leverrier2022quantum}. This implies that for any sequence of local errors whose energy penalty is a sufficiently small constant fraction of the energy barrier $\Delta_{\text{QTC}}$, the weight of the error is sufficiently small to be corrected by the decoder in Refs.~\cite{Leverrier2022Efficient,Dinur2022Good,Gu2022An,Leverrier2022A}. That is, \textbf{QTCD} achieves a constant fraction of the energy barrier $m_{\text{QTC}}\geq  \Omega( \Delta_{\text{QTC}})$. For good Quantum Tanner Codes $\Delta_{\text{QTC}}=\Omega(n)$, and hence there exists a constant $\mu>0$ such that $m_{\text{QTC}}\geq \mu n$. 
    Combining this with the bound in the paragraph above, we find that the concatenated matching decoder achieves a constant fraction of the energy barrier $m\geq \frac{4 \mu}{w w'} L$ when applied to a sequence of local Pauli operators on Layer Codes based on a family of Good Tanner Codes. 
\end{proof}

In the appendix we show that the concatenated matching decoder corrects up to a constant fraction of the layer code distance. 

\textit{Discussion.}---In this work we have demonstrated that Layer Codes based on a family of good Quantum Tanner Codes are partially self-correcting while weakly coupled to a thermal bath, at a sufficiently low temperature, for some time followed by decoding with a concatenated matching decoder. Furthermore, we established that Layer Codes exhibit a strong form of partial self-correction which protects $\Omega(L)$ logical qubits for a time $\Omega(e^{a c \beta  L})$ up to a temperature dependent system size $L_\bullet= O(e^{\frac{1}{3} (1-a) \beta})$. 
Our key technical result was a proof that the concatenated matching decoder achieves a constant fraction of the energy barrier. 
In the appendix we also show that this decoder achieves a constant fraction of the distance. 

Several future directions for extensions of this work present themselves.
First, the memory time is potentially longer than our predictions as we have used an upper bounded on the storage error. While we do not expect true self-correction~\cite{Baspin2025The}, what is the best achievable scaling of the memory time and maximum system size for partial self-correction?
Second, our decoder is fast, with expected sublinear time complexity. It is an interesting challenge to prove this, potentially after replacing minimum-weight perfect matching with a provably fast decoder~\cite{Delfosse2021Almost}. 
This would address an open question from Ref.~\cite{Eggerickx2025Almost}.
Third, numerical simulations are desirable to paint a more detailed picture of our partial self-correction results. However, this is a challenge due to the large size of the input good Quantum Tanner Code families. One approach is to consider families of good qLDPC codes with smaller instances as inputs to the Layer Code construction. A concurrent work, Ref.~\cite{Gu2025Layer}, instead focused on simulating Layer Codes based on random good dense code inputs. The flexibility of these codes allowed the authors of Ref.~\cite{Gu2025Layer} to perform numerical simulations, which produced results consistent with their analytical demonstrations of partial self-correction. 

The problem of finding a self-correcting memory at finite temperature in three dimensions remains open. If such a memory can protect many logical qubits, as our partially self-correcting Layer Code memory does, it is natural to also ask: Is a universal self-correcting quantum computer at finite temperature possible in three-dimensions~\cite{Bombin2013Self}?

\vspace{.5cm}
\textit{Acknowledgements.}---  
The author acknowledges inspiring discussions with Shouzhen (Bailey) Gu and Libor Caha on their related work which appeared in a recent arxiv posting~\cite{Gu2025Layer}. 
The author also acknowledges productive collaboration with Nou\'edyn Baspin during the early stages of this project. 
Part of this work was done during the ``Logical Gates for Encoded Qubits'' workshop at the Yukawa Institute for Theoretical Physics, and the ``Noise-robust Phases of Quantum Matter'' program at the Kavli Institute for Theoretical Physics. 
DJW is supported by the Australian Research Council Discovery Early Career Research Award (DE220100625). This research was supported in part by grant NSF PHY-2309135 to the Kavli Institute for Theoretical Physics (KITP).

\bibliography{references.bib}

\appendix

\section{Review of the general setting}

In this section we introduce some background material about the general setting for this work, following Ref.~\cite{bravyi2013classification}.

The standard Hamiltonian associated to a CSS stabilizer code $\mathcal{C}$ is
\begin{align}
    H_{\mathcal{C}} = - \sum_{i} \frac{1}{2}(\idty-c_i^X) - \sum_{j} \frac{1}{2}(\idty-c_j^Z),
\end{align}
where the sums run over all $X$-type code checks $c_i^X$ and all $Z$-type code checks $c_j^Z$. 
The Hamiltonian is normalized to have the code-space of $\mathcal{C}$ as its zero-energy ground-space. 
Excitations of the Hamiltonian each cost one unit of energy, and correspond to syndromes of the code $\mathcal{C}$. We remark that the overall energy scale of the Hamiltonian can be absorbed into a redefinition of the temperature, and so we set it to one for simplicity. 
The degeneracy of the ground space is $2^k$, and the same is true for all valid syndrome configuration spaces, i.e. those satisfying all meta-checks. 
All families of topological codes (a code with local stabilizer generators that corrects all local errors) give rise to reasonable Hamiltonians that are stable to local perturbations~\cite{Bravyi2010Stability,Bravyi2011A,Lavasani2024On,DeRoeck2024LDPC,Yin2025Low}, this includes the Layer Codes. 

The energy penalty of a Pauli operator $P$ is the number of syndromes violated by that operator. 
For a sequence of Pauli-$X$ operators $\{P_0P_1\cdots P_j\}_{j=0,\dots,\tau}$ the energy penalty is the maximum energy penalty for any intermediate operator $P_0P_1\cdots P_j$.
The energy barrier of an operator $E$ is the minimum energy penalty of any sequence that satisfies $E=P_0P_1\cdots P_\tau$. 
The energy barrier of a code is the minimum energy barrier of any logical operator representative in that code.
For any topological code defined on a $D$-dimensional hypercubiod lattice of size $\Theta(L^D)$ the energy barrier is upper bounded by $O(L^{D-1})$, here $L$ is known as the linear system size. This is because all logical operators can be cleaned into codimension-1 regions~\cite{bravyi2009no}. 

In this work we follow Ref.~\cite{Bravyi2013} and consider dynamics of an initial state $\rho(0)$ in the ground space of a stabilizer Hamiltonian, which then evolves under that Hamiltonian coupled to a thermal bath in the Davies weak coupling limit. This results in a Markovian master equation
\begin{align}
    \frac{d}{d t}{\rho}(t) = -i[H,\rho(t)]+\mathcal{L}\big(\rho(t)\big), 
\end{align}
for $t\geq 0,$ where $\mathcal{L}$ is the Lindblad generator. 
The superoperator $\mathcal{L}$ is defined by Hermitian operators $\{ A_\alpha\}$ that describe the coupling of the system to the bath. In local many-body quantum dynamics, as considered in this work, the operators $A_\alpha$ are assumed to have support on a constant number of qubits. 
A simple example is a quantum analog of Glauber dynamics in which the system-bath coupling operators are taken to be the set of all single qubit Pauli $X$ and $Z$ operators. 

The system-bath coupling operators are expanded into Bohr frequencies $A_\alpha = \sum_{\omega} A_{\alpha,\omega}$, where $A_{\alpha,\omega}$ is the component of $A_\alpha$ that maps eigenvectors with energy $E$ to eigenvectors with energy $E-\omega$. We now write the Lindblad generator in this basis 
\begin{align}
    \mathcal{L}(\rho) = \sum_{\alpha,\omega} h(\alpha,\omega)\big(A_{\alpha,\omega} \rho A_{\alpha,\omega}^\dagger - \frac{1}{2} \{ \rho, A_{\alpha,\omega}^\dagger A_{\alpha,\omega} \} \big), \nonumber
\end{align}
where $h(\alpha,\omega)$ describes the rate of quantum jumps caused by $A_\alpha$ that transfer energy $\omega$ from the system to the bath. 
For a thermal bath at inverse temperature $\beta$, the rates must satisfy detailed balance
\begin{align}
    h(\alpha,-\omega)=e^{-\beta \omega} h(\alpha,\omega),
\end{align}
This ensures that the thermal Gibbs state is a steady state, $\mathcal{L}(\frac{1}{Z}e^{-\beta H})=0$, which is unique under an ergodicity condition.

The results of Ref.~\cite{Bravyi2013}, and our results here,  require $\| A_\alpha \| \leq 1$ and $\max h(\alpha,\omega) =O(1).$  They further require an upper bound on $f$ the energy penalty of the Pauli operators that appear in a Pauli basis expansion of the jump operators. Here, $f$ is constant, $f=O(1)$,  as each $A_{\alpha,\omega}$ is local, which follows from the locality assumed of the $A_{\alpha}$ operators and the checks in the stabilizer Hamiltonian. 
We have suppressed this subtlety in the main text for simplicity of presentation. Technically, the required decoder energy barriers we have quoted in the main text should be increased by an $O(1)$ additive constant to $m+2f$. This does not change their scaling behaviours. 

The action of the measurement and decoding procedure that is applied at the end of the quantum memory's time evolution can be described by a completely positive trace preserving map of the form
\begin{align}
    \Phi_{ec} = \sum_S R(S) P_{ec}(S) \Pi_S \rho \Pi_S P_{ec}(S). 
\end{align}
Here, the sum is over all syndromes $S$. The operator $\Pi_S$ is the projector onto the $S$ syndrome subspace, and $P_{ec}(S)$ is the Pauli correction operator returned by the decoder applied to syndrome $S$. In Algorithm~\ref{alg:CMD} $P_{ec}(S)$ corresponds to the variable $R$. 

Finally, we explain the derivation of Eq.~\eqref{eq:prePSC} in more detail. 
In Section~V.B of Ref.~\cite{Bravyi2013} a bound on the memory error is derived. We do not reproduce the proof here. This leads to the first inequality below. 
\begin{align}
    \varepsilon (t) :&= \| \rho(0) - \Phi_{\text{ec}}\big(\rho(t)\big) \|
    \\
    &\leq O(tN) \text{tr} ( Q_m e^{-\beta H} ) \\
    &\label{eq:App2}
    \leq O(tN) 2^k \sum_{n\geq m}  
    \begin{pmatrix} N \\ n \end{pmatrix} e^{-\beta n} \\
    &\label{eq:App3}
    \leq O(tN) 2^k e^{- a \beta m} \sum_{n\geq m}  
    \begin{pmatrix} N \\ n  \end{pmatrix} e^{-\beta n+a\beta m} \\
    &\label{eq:App4}
    \leq O(tN) 2^k e^{- a \beta m}\sum_{n\geq m} 
    \begin{pmatrix} N \\ n \end{pmatrix} e^{-(1-a)\beta n} \\
    &\label{eq:App5}
    \leq O(tN) 2^k e^{- a \beta m} ( 1+ e^{-(1-a)\beta} )^N .
\end{align}
To reach Eq.~\eqref{eq:App2} we note that there are $N$ possible syndrome locations, we expand the trace out into a sum over energy eigenspaces of $H$ with energy at least $m$. The eigenspace with energy $n\geq m$ contains at most $\begin{pmatrix} N \\ n \end{pmatrix}$ distinct syndrome subspaces, each with degeneracy $2^k$. The trace of $e^{-\beta H}$ on the eigenspace with energy $n$ is simply $e^{-\beta n}$. The remaining steps are simple manipulations, using $n\geq m$, and noting the sum in Eq.~\eqref{eq:App4} is a partial expansion of a binomial to the $N$th power which is upper bounded by the full binomial to the $N$th power.

\section{Concatenated Matching Decoder Distance}

In this section we show that the concatenated matching decoder achieves a constant fraction of the layer code distance. For this we make a small modification to the Layer Code construction, which produces equivalent codes. We consider \textit{extended Layer Codes} where all blue $X$-check layers and red $Z$-check layers are extended to the boundaries of the hypercuboid spanned by the grey data layers. This is achieved by expanding the trivalent defects where blue (red) layers end into a four-valent defect and a smooth (rough) boundary and then stretching the boundary to the edge of the hypercubiod, see Fig.~\ref{fig:Boundaries}. This transformation introduces additional qubits, but  does not change the parameter scaling or the Layer Codes. 
\begin{lemma}[Concatenated matching decoder distance]
\label{Lem:2}
    For extended Layer Codes based on a family of good Quantum Tanner Codes, Algorithm~\ref{alg:CMD} decodes any Pauli error with weight up to a constant fraction of the code distance. That is, there exists a constant $\alpha>0$ such that the decoder achieves distance $\alpha L^2$. 
\end{lemma}
\begin{proof}
There are two inequivalent logical classes for \textbf{MWPM} on the blue layers in stage 1 of Algorithm~\ref{alg:CMD}~which differ by flipping the parity of the matching to the smooth boundaries. 
One of these matchings produces a correction that neutralizes the number of syndromes entering the $X$-check layer from each adjacent data layer, the other has an odd parity of syndromes leaving the $X$-check layer into each adjacent data layer. 
The errors within an $X$-check layer must have weight $\geq L/2$ to cause the latter type of correction to occur. 
In this case, it is possible that the residual errors introduced to the adjacent data layers by the correction correspond to syndromes of an error in those data layers with weight $\geq L/2$ which takes the form of a partial stabilizer of the input code, i.e.~an error on up to $w/2$ of the qubits of the input code. 
In the former case, the multiplicative increase in the effective weight of the error on the input code is upper bounded by a factor of $w/2$. This is because, in the worst case, error strings which correspond to a partial stabilizer of the input code of weight up to $w/2$ can be cleaned into a single string inside an overlapping $X$-check layer. 

When \textbf{MWPM} is applied to a grey layer, it can only produce a nontrivial spanning error if the original error has weight greater than $L/2$. This is in analogy to standard minimum-weight perfect-matching decoding of the surface code. 
To cause a logical error of the Layer Code under optimal decoding, the number of grey and blue layers with spanning errors at stage 3 of Algorithm~\ref{alg:CMD} must be greater than $d_{\text{input}}/2$. 
To cause a logical error of the Layer Code with a decoder that achieves a fraction of the distance $ d_{\text{input}}/f$, the number of layers with spanning errors must be $ d_{\text{input}}/f$. 
Combining the number of Layers that must support a spanning error $d_{\text{input}}/2$, with the weight of a spanning error $L/2$, and the potential distance reduction factor $2/w$ establishes that the iterated matching decoder achieves distance $d_{\text{input}}L/(2 w f)$. 
Interestingly, if the input decoder achieves the optimal $d_{\text{input}}/2$, then the iterated decoder achieves a quarter of the lower bound on the Layer Code distance $d_{\text{input}}L/(2 w)$

For a family of good Quantum Tanner Codes, there exists a constant $\delta$ such that \textbf{QTCD} decodes up to a distance $d_{\text{QTCD}} \geq \delta n$, see Refs.~\cite{Leverrier2022Efficient,Dinur2022Good,Gu2022An,Leverrier2022A}. 
Hence, The concatenated matching decoder described in Algorithm~\ref{alg:CMD} achieves distance $d_{\text{CMD}}\geq \frac{\delta}{2w} L^2$. 
\end{proof}

\end{document}